\documentclass[a4paper,11pt]{article}
\usepackage{a4wide}

\usepackage[margin=1in]{geometry}
\usepackage[T1]{fontenc}
\usepackage[english]{babel}
\usepackage{color}
\usepackage[textsize=small]{todonotes}
\usepackage[algo2e,ruled,lined,linesnumberedhidden,longend]{algorithm2e}
    \SetArgSty{textsl}
    \SetCommentSty{textnormal}
    \SetKwFor{For}{for}{}{endfor}
\usepackage[cmex10]{amsmath}
\usepackage{amssymb}
\usepackage{amsfonts}
\usepackage{amsthm}
\usepackage{enumerate}
\usepackage{multicol}
\usepackage{multirow}
\usepackage{lmodern}
\usepackage{enumitem}
\usepackage{mathrsfs}

\usepackage{caption}
\addto\captionsenglish{}

\SetEnumitemKey{ncases}
  {itemsep=0pt,align=left,leftmargin=\parindent,
   itemindent=!,label={\normalfont\textbf{Case $\arabic*$}}}

\SetEnumitemKey{cases}
  {itemsep=0pt,align=left,leftmargin=\parindent,
   itemindent=!,before=}

\theoremstyle{plain}
\newtheorem{theorem}{Theorem}
\newtheorem{proposition}[theorem]{Proposition}

\newtheorem{lemma}[theorem]{Lemma}
\newtheorem{notation}[theorem]{Notation}

\theoremstyle{conjecture}

\theoremstyle{definition}
\newtheorem{definition}{Definition}

\theoremstyle{remark}
\newtheorem{remark}{Remark}

\usepackage{algorithmic}

\newcommand{\eqdef}{\stackrel{\text{def}}{=}}

\newcommand{\F}{\ensuremath{\mathbb{F}}}

\newcommand{\prob}{\ensuremath{\textsf{prob}}}
\newcommand{\code}[1]{\ensuremath{\mathscr{#1}}}

\newcommand{\Code}[1]{\code{#1}}

\newcommand{\word}[1]{\mathbf{#1}}

\newcommand{\pv}{\word{p}}
\newcommand{\qv}{\word{q}}
\newcommand{\rv}{\word{r}}
\newcommand{\uv}{\word{u}}
\newcommand{\vv}{\word{v}}

\newcommand{\xv}{\word{x}}
\newcommand{\yv}{\word{y}}

\newcommand{\mat}[1]{\ensuremath{\boldsymbol{#1}}}

\newcommand{\Dm}{\mat{D}}

\newcommand{\Gm}{\mat{G}}
\newcommand{\Hm}{\mat{H}}

\newcommand{\Yc}{{\mathcal Y}}

\newcommand{\fq}{\F_{q}}

\newcommand{\UV}{\left(U\mid U+V\right)}

\newcommand{\gUV}[2]{\left(#1\mid #1+#2\right)}
\newcommand{\uuv}{\left(\mathbf u \mid \mathbf u+\mathbf v\right)}
\newcommand{\guv}{\left( \mathbf u \Dm_1 + \mathbf v \Dm_2 \mid \mathbf u \Dm_3 +\mathbf v \Dm_4 \right)}
\newcommand{\AB}{\left(U_1\mid U_1+V_1 \mid U_1+U_2 \mid U_1+U_2+V_1+V_2\right)}
\newcommand{\Oq}{\mathcal O\left(\frac{1}{q}\right)}

\newcommand{\esp}{{\mathbb{E}}}
\newcommand{\var}{{\mathbf{Var}}}
\newcommand{\norm}[1]{\left|\!\left|#1\right|\!\right|}

\newcommand{\oq}{\mathcal O\left(\frac{1}{q}\right)}

\usepackage{authblk}

\begin{document}

\title{Using Reed-Solomon codes in the $\UV$ construction and an application to cryptography}

\author{Irene M{\'a}rquez-Corbella\footnote{Inria,  Email: \texttt{irene.marquez-corbella@inria.fr}.},\quad Jean-Pierre Tillich\footnote{Inria, Email: \texttt{jean-pierre.tillich@inria.fr}.}}

\maketitle
\begin{abstract}
In this paper we present a modification of Reed-Solomon codes that beats the Guruwami-Sudan $1-\sqrt{R}$ decoding radius 
of Reed-Solomon codes at low rates $R$.  
The idea is to choose Reed-Solomon codes $U$ and $V$ with appropriate rates in a $\UV$ construction and
to decode them with the Koetter-Vardy soft information decoder.
We suggest to use a slightly more general version of these codes (but which has the same decoding performances
as the $\UV$-construction) for code-based cryptography, namely to build a McEliece scheme.
The point is here that these codes not only perform nearly as well (or even better in the low rate regime) as
Reed-Solomon codes, their structure seems to avoid the Sidelnikov-Shestakov attack which broke
a previous  McEliece proposal based on generalized Reed-Solomon codes.
\end{abstract}

\section{Introduction}\label{sec:introduction}
\paragraph{Improving upon the error correction performance of RS codes.}
Reed-Solomon(RS) codes are among the most extensively used error correcting codes. 
It has long been known how to decode them
up to half the minimum distance. This gives a decoding algorithm that 
corrects a fraction $\frac{1-R}{2}$ of errors in an RS code of  rate $R$. However, it is only in the late nineties that a breakthrough was obtained
in this setting with Sudan's algorithm \cite{S97b} and its improvement in \cite{GS99} who showed how to go beyond this barrier with an algorithm 
which in its \cite{GS99} version decodes any fraction of errors smaller than $1 - \sqrt{R}$. 
Later on, it was shown that this decoding algorithm could also be modified a little bit in order to cope with soft information on the errors \cite{KV03}. Then it was  realized 
in \cite{PV05} 
that by a slight modification of RS codes and by an increase of the alphabet size it was possible to beat the $1-\sqrt{R}$ decoding radius. 
Their new family of codes is  list decodable beyond this radius for low rate.
Then, 
\cite{GR06} improved on these codes by presenting a new family of codes, namely {\em folded RS codes} 
with a polynomial time decoding algorithm achieving the list decoding capacity $1-R-\epsilon$ for every rate $R$ and $\epsilon >0$.

The first purpose of this paper is to present another modification of RS codes that improves the fraction of errors
that can be corrected.  It consists in using RS codes in a $\UV$ construction. 
 We will show that, in the low rate regime, this class of codes outperforms rather significantly a classical RS code decoded with the Guruswami and Sudan algorithm 
\cite{GS99}. The point is that this $\UV$ code can be decoded in two steps : 
\begin{enumerate}
\item First by subtracting the left part $y_1$ to the right part $y_2$ of the received vector $\left(y_1\mid y_2\right)$ 
and decoding it with respect to $V$. In such a case, we are
left with decoding a RS code with about twice as many errors. 
\item Secondly, once we have 
recovered the right part $v$ of the codeword, we can get a word $\left(y_1 \mid y_2-v\right)$ which should
match two copies of a same word $u$ of $U$. 
We can model this decoding problem by having some
soft information. 
\end{enumerate}
It turns that the last channel error model is much less noisy than the original $q$-ary symmetric channel we started with. This soft information can be used in Koetter and Vardy's decoding algorithm.
By this means we can choose $U$ to be a RS code of much bigger rate than $V$. All in all, it turns out that by choosing $U$ and $V$ with appropriate rates we can beat the $1 - \sqrt{R}$ bound in the low-rate regime.

It should be noted however that beating this $1- \sqrt{R}$ bound comes at the cost of having now an algorithm which does not work as for the aforementioned
papers \cite{S97b,GS99,PV05,GR06} for every error of a given weight (the so called adversarial error model) but with probability $1-o(1)$ for errors of a given weight. However, contrarily to \cite{PV05,GR06} which results in a significant increase of the alphabet size of the code, our alphabet size actually 
decreases when compared to a RS code: it can be half of the code length and can be even smaller when we apply this construction recursively. Indeed, we will show
that we can even improve the error correction performances by applying this construction again to the $U$ and $V$ components, i.e  
 we can choose $U$ to be a $(U_1|U_1+V_1)$ code
and we replace in the same way the 
RS code $V$ by a $(U_2|U_2+V_2)$
where $U_1, U_2, V_1, V_2$ are RS codes.

\paragraph{Application to cryptography.} 
In a second part of the paper we show how to use such codes (or codes derived by this approach) for cryptographic purposes, i.e. in a McEliece cryptosystem \cite{M78}. Recall that this public-key cryptosystem becomes more and more fashionable due to the threats on the most popular public key cryptosystems used todays, namely RSA or DASA and ECDSA that would
be completely broken by Shor's algorithm \cite{S94a} if a large scale quantum computer could be built.
Indeed, it is unlikely that a quantum computer would be able to threaten the security of the McEliece scheme because it is based on an 
$\mathrm{NP}$-complete problem, namely  decoding a linear code.

Probably one of the main drawback of McEliece when compared to RSA, DSA or ECDSA is its rather large key size.
There have been several attempts to decrease the key size either by moving to more structured codes 
or to codes which have better error correction radius \cite{N86,BBCRS14}.
 Many of the structured algebraic proposals have been broken (see for instance \cite{FOPPT16}) but some of the quasi-cyclic code families
that rely on modified LDPC codes or MDPC codes \cite{BBC13,MTSB13} seem to resist cryptanalysis up to now.
Relying on codes with better decoding performance met a similar fate, since here again many proposals of this kind have been broken.
For instance \cite{N86} suggests to replace the binary Goppa codes of the original McEliece cryptosystem by Generalized RS codes (GRS) because of their much better decoding performance, but
it got broken in \cite{SS92}.

There have been several attempts to repair GRS codes in this context either by adding random columns to the generator matrix of a GRS code \cite{W06}
or by multiplying this generator matrix by the inverse of a sparse matrix with small average row weight $m$ \cite{BBCRS14}. 
The \cite{W06} attempt got broken in \cite{CGGOT14} and the parameters of \cite{BBCRS14} got broken in \cite{COTG15} because $m$ was chosen to 
be too small.
The problem with the \cite{BBCRS14} approach is that the attack of \cite{COTG15} fails when $m=2$, but the solution is then 
no more competitive when compared to a Goppa code because the decoding radius gets also scaled down by a multiplicative factor of $m$ when 
compared to a GRS code.

We suggest here to revive the \cite{BBCRS14} approach with a generalized $\UV$ scheme based on RS codes that has basically the same
decoding capacitiy as a RS code and that looks in many respects like the \cite{BBCRS14} scheme with $m=2$,
This approach is also related to the approach pioneered by Wang in \cite{W15}. His code can be viewed as a certain 
subcode of our $\UV$ construction. However, the decrease of the code rate results in a significant deterioration
of the key size when compared to a code with the same error correction capacity as an RS code.


\begin{notation}
Throughout the paper we will use the following notation.
\begin{itemize}
\item  A linear code of length $n$, dimension $k$ and distance $d$ over a finite field $\mathbb F_q$ is refered to as an $[n,k,d]_q$-code. We denote its dual by $\mathcal C^{\perp}$ which is defined to be an $[n,n-k, d^{\perp}]_q$ code.
\item The Hamming weight of a vector $\mathbf x$, denoted by $\mathrm w_H(\mathbf x)$, is defined as the number of nonzero elements.
\item For a vector $\xv$ we denote by $x(i)$ the $i$-th coordinate of $\xv$.
\end{itemize}
\end{notation}

\section{$\UV$ construction}
\label{Section3}
In this section, we recall a few facts about the $(U|U+V)$ construction and their decoding.

\begin{definition}
Let $U$ be an $[n,k_u,d_u]_q$ code and $V$ be an $[n,k_v,d_v]_q$ code.
We define the $\left(U\mid U+V\right)$-construction of $U$ and $V$ as the linear code:
$$\mathcal C = \left\{ (\mathbf u\mid\mathbf u + \mathbf v)\mid \mathbf u \in U \hbox{ and } \mathbf v \in V\right\}.$$
The code $\mathcal C$ has parameters $[2n, k_u + k_v, \min \left\{ 2d_u, d_v\right\}]_q$.
A generator matrix of $\mathcal C$ is:
$$\left( \begin{array}{c|c}G_u & G_u\\ \hline\mathbf 0 &G_v\end{array}\right)\in \mathbb F_q^{(k_u+k_v)\times 2n}$$
where $G_u$ and $G_v$ are generator matrices of $U$ and $V$ respectively.
\end{definition}

\subsection{Soft-decision decoding of $\UV$ codes}
\label{Decoding-Algo}
Let $U$ and $V$ be two codes  with parameters
$[n,k_u,d_u]_q$ and $[n,k_v,d_v]_q$, respectively
and $\mathcal C \eqdef \UV$.
Suppose we transmit the codeword $\uuv\in \mathcal C$ over a noisy channel and we receive the vector:
$\yv = (\mathbf y_1 \mid \mathbf y_2) = \uuv + (\mathbf e_1 \mid \mathbf e_2)$. 

Decoding proceeds in two steps:
\begin{enumerate}
\item We combine $\mathbf y_1$ and $\mathbf y_2$ to find $\mathbf v$. That is, we decode $\mathbf y_1 - \mathbf y_2 = \mathbf v + \mathbf e_2 - \mathbf e_1$ with respect to $V$. In the case of a soft decoder for $V$ we compute first the probability 
$$\prob(v(i)=\alpha|y_1(i),y_2(i)) ~~\hbox{ for all }~~\alpha \in \fq.$$
\item We subtract  $(\mathbf 0\mid \mathbf v)$ to $(\mathbf y_1 \mid \mathbf y_2)$ to get $(\mathbf u + \mathbf e_1 \mid \mathbf u + \mathbf e_2) = (\mathbf z_1 \mid \mathbf z_2)$. This is a noisy version of $(\uv \mid \uv)$. 
We compute now for all $\alpha \in \fq$ and all coordinates $i$ the probabilities $\prob(u(i)=\alpha|z_1(i),z_2(i))$
which is then passed to a soft decoder for $U$.
\end{enumerate}

Let us explain how these probabilities can be computed. 
We assume that the noise model is given by a discrete memoryless channel 
with input alphabet  $\mathbb F_q$
and output alphabet $\mathcal Y$. The received vector is denoted by $\mathbf y=(y(1), \ldots, y(n))\in \mathbb \Yc^n$ 
and the channel model specifies the transition probabilities with  the following matrix $\Pi_{\mathbf y}$ 
$$\Pi_{\mathbf y}^i (\alpha) = \prob (\alpha \mid y(i)) ~\hbox{ for } i=1, \ldots, n \hbox{ and } \alpha \in \fq.$$
$\Pi_{\mathbf y}^i$ denotes here the $i$-th column of $\Pi_{\mathbf y}$ and  $\Pi_{\mathbf y}^i (\alpha)$  refers to the entry in the $i$-th column and  row indexed by $\alpha \in \mathbb F_q$.

We will refer to  $\Pi$ as the $q\times n$ \textbf{reliability matrix} of the codewords symbols.
We will see below that this reliability matrix can also be obtained through the $\UV$ decoding process.
We will particularly be interested here in the $q$-ary symmetric channel model.

This channel is parametrized by the crossover probability $p$ and the channel will be denoted here by $q\hbox{-SC}_{p}$. Here each time an element from $\mathbb F_q$ is transmitted and is received either the unchanged input symbol, with probability $1-p$, or any of the other $q-1$ symbols, with probability $\frac{p}{q-1}$. In other words, 
the reliability matrix $\Pi_{\mathbf y}$ for $q\hbox{-SC}_{p}$ is defined as follows:
$$\Pi_{\mathbf y}^i (\alpha) = \prob\left(\alpha\mid y(i)\right) = \left\{ \begin{array}{ll}
1-p & \hbox{ if } \alpha = y(i)\\
\frac{p}{q-1} & \hbox{ if } \alpha \neq y(i)
\end{array}\right.$$

Thus, all columns of $\Pi_{\mathbf y}$ are identical up to permutation:
$$\Pi_{\mathbf y}^i =  \left(\begin{array}{c}
1-p \\ \frac{p}{q-1} \\ \vdots \\ \frac{p}{q-1}
\end{array}\right) \hbox{ (up to permutation)} $$
with $i=1, \ldots, n$.

Let us recall now how the reliability matrices for the decoder of $U$ and $V$ are computed from the initial reliability matrix.

\paragraph{ Reliability matrix for the $V$-decoder} We call in what follows the error model for the $V$-decoder the {\bf sum model}
and denote the associated reliability matrix by $\Pi \oplus \Pi$ when $\Pi$ is the initial reliability matrix.
Recall that before  decoding, for each symbol $X$ of $V$ that we want to decode we subtract  two symbols $X_1$ 
and $X_2$ of the $\UV$ code:
$$
X = X_2-X_1.
$$
For each of these symbols we have a reliability information $\prob(X_1=\alpha \mid Y_1)$ and 
$\prob(X_2=\beta \mid Y_2)$ where $Y_1$ and $Y_2$ are random variables that are initially the received symbols corresponding to
$X_1$ and $X_2$ after 
transmission on the noisy channel but that become sets of received symbols when we iterate the $\UV$ construction as will be seen. 
When $X_1$ and $X_2$ are uniformly distributed it can be verified that

\begin{equation*} 
\prob(X=\alpha|Y_1,Y_2)=\sum_{\beta \in \mathbb F_q} \prob(X_1=\beta |Y_1)\cdot \prob(Y_2=\alpha+\beta | Y_2)
\end{equation*}

\normalsize This leads to the following definition.
$$(\Pi\oplus \Pi)_{\mathbf y}^i (\alpha) \eqdef \sum_{\beta \in \mathbb F_q} \Pi_{\mathbf y_1}^i (\beta) \cdot \Pi_{\mathbf y_2}^i (\alpha + \beta)$$
where $\yv_1$ and $\yv_2$ are the realizations of the channel transmission of $u$ and $u+v$ respectively.
We also denote by $\Pi_{\mathbf y_1}^i \oplus \Pi_{\mathbf y_2}^i$, where each element represents a column vector,  the $i$-th column of the $\Pi \oplus \Pi$ matrix.

\paragraph{Reliability matrix for the $U$-decoder} The computation of $\prob(u(i)=\alpha|z_1(i),z_2(i))$ can be performed 
by computing the probability that a uniformly distributed random variable over $\fq$ is equal to $\alpha$ given 
two received symbols $y_1$ and $y_2$ for $X$ sent over two memoryless channels (and which are  chosen uniformly at random in $\fq$). This probability is readily seen to be 
equal to
$$
\prob(\mathcal X=\alpha \mid y_1 \hbox{ and } y_2) = 
\frac{\prob(X=\alpha \mid y_1) \cdot \prob (X=\alpha \mid y_2)}{\sum_{\beta\in \mathbb F_q} \prob(X=\beta \mid y_1) \cdot \prob ( X=\beta \mid y_2)}$$

We denote by $\Pi \times \Pi$ the reliability matrix (the input) to a soft-decision decoding algorithm for the code $U$.  
Thus, each element of the reliability matrix $\Pi \times \Pi$ related to the aforementioned 
quantities $\yv$ and $\vv$  is defined by:
$$( \Pi\times \Pi)_{\mathbf y, \vv}^i(\alpha) = \frac{\Pi_{\yv_1}^i (\alpha) \cdot \Pi_{\mathbf \yv_2}^i (\alpha+v(i))}{\sum_{\beta \in \mathbb F_q} \Pi_{\mathbf \yv_1}^i  (\beta ) \cdot \Pi_{\mathbf \yv_2}^i(\beta+v(i))}.$$
To simplify notation we will generally avoid the dependency on $\vv$ and simply write $( \Pi\times \Pi)_{\mathbf y}$ .

\subsection{Algebraic-soft decision decoding of RS codes}

Let us recall how the Koetter-Vardy soft decoder \cite{KV03} can be analyzed.
 By \cite[Theorem 12]{KV03} their decoding algorithm outputs a list that contains the codeword $\mathbf c\in C$ if

$$\frac{\left\langle \Pi, \lfloor \mathbf c\rfloor \right\rangle}{\sqrt{\left\langle \Pi, \Pi\right\rangle}} \geq \sqrt{k-1}+o(1)$$
as the codelength $n$ tends to infinity,
where $\lfloor \mathbf c\rfloor$ represents a $q\times n$ matrix with entries $c_{i,\alpha} = 1$ if $c_i = \alpha$, and $0$ otherwise; and
$\left\langle A, B\right\rangle$.
 denotes the inner product of the two $q\times n$ matrices $A$ and $B$, i.e. 
$$\left\langle A, B\right\rangle \eqdef \sum_{i=1}^q\sum_{j=1}^n a_{i,j}b_{i,j}.$$
The algorithm uses a parameter $s$ (the total number of interpolation points counted with multiplicity). The Little-O $o(1)$ depends on 
the choice of this parameter and the parameters $n$ and $q$. 
It can be chosen as a function of $q$ in such a way that, when $\left\langle \Pi, \Pi\right\rangle$
has a lower bound given by some positive constant then, the Little-O of this formula is bounded from above
by a function of $q$ that goes to $0$ as $q$ goes to infinity. 
We will consider here only discrete symmetric channel models that are defined below. Let us first introduce some notation.
\begin{notation}[Probability error vector of a Discrete Memoryless Channel (DMC)]
For a given DMC with $q$-ary inputs we denote by  $\pi$  the probability vector
$\pi=(\prob(x=\alpha|y))_{\alpha \in \fq}$ where $x$ is the symbol that has been sent through the channel and $y$ is the received 
symbol.  For a vector $\xv = (x(\beta))_{\beta \in \fq}$ we denote
by $\xv^{+\alpha}$ the vector $\xv^{+\alpha}=(x(\beta+\alpha))_{\beta \in \fq}$.
\end{notation}

By viewing $\pi$ as a random variable 
(namely as a function of the random variable $y$), 
we define as in \cite{BB06}
a symmetric channel by
\begin{definition}[discrete symmetric channel with $q$-ary inputs]
A DMC with $q$-ary inputs is said to be symmetric if and only if for any $\alpha$ in $\fq$ we have
\begin{equation}
\label{eq:symmetry} 
p(\alpha) \prob(\pi = \pv) = p(0) \prob(\pi = \pv^{+\alpha}).
\end{equation} 
\end{definition}
Note that this implies that in a discrete symmetric channel, for any possible realization $\pv$ of the probability vector $\pi$ 
(i.e. when $\prob(\pi=\pv) \neq 0$) 
we necessarily have $p(0)\neq 0$, since otherwise we would have $p(\alpha)=0$ for all $\alpha \neq 0$, a contradiction with the 
fact that $\pv$ is a probability vector $\sum_{\alpha \in \fq} p(\alpha) = 1$.
It is proved in \cite{BB06} that symmetric channels are closed under the $\oplus$ and $\times$ operations on channels
defined in Subsection \ref{Decoding-Algo}.
We give now the asymptotic behavior for a symmetric channel of 
the Koetter-Vardy decoder, but before doing this we will need a few lemmas. 

\begin{lemma}\label{lem:simple_formula}
Let $\pi=(\pi(\alpha))_{\alpha \in \fq}$ be the probability vector associated to a discrete symmetric channel with $q$-ary inputs. Then 
$$
\esp(\pi(0)) = \esp\left(\norm{\pi}^2\right),
~~~~\hbox{ with } \norm{\pi}^2 \eqdef \sum_{\alpha \in \fq} \pi(\alpha)^2.$$
\end{lemma}

\begin{proof}
Let us begin by observing that for a symmetric channel if $\pv$ is a possible realization of the probability vector $\pi$, then
$\pv^{+\alpha}$ is also a possible realization of this probability vector as soon as $p(\alpha) \neq 0$. This motivates to introduce
the equivalence relation between probability vectors over $\fq$ :
$$\pv \equiv \qv \hbox{ iff there exists }\alpha \in \fq \hbox{ for which }\pv^{+\alpha}=\qv.$$
We consider for our DMC the set of all equivalence classes of the probability vector $\pi$ and
denote by $\mathcal R$ a set of representatives of such equivalence classes. For a representative $\rv$ of an equivalence
class we denote by $c(\rv)$ the class to which it belongs. Observe now that for such a representative $\rv=(r(\beta))_{\beta \in \fq}$ we have 
\begin{equation}
\prob\left(\pi=\rv^{+\alpha}|\pi \in c(\rv)\right)  =  \sum_{\alpha \in \fq } \frac{K r(\alpha)}{r(0)}
\end{equation}
for some constant $K>0$ by using \eqref{eq:symmetry}.
Since
$$\begin{array}{ccc}
\displaystyle
\sum_{\alpha \in \fq} \prob\left(\pi=\rv^{+\alpha}|\pi \in c(\rv)\right) = 1 &
\hbox{ and }&
\displaystyle
\sum_{\alpha \in \fq} r(\alpha)  =  1
\end{array}$$
we necessarily have $K=r(0)$. Therefore,
\begin{eqnarray*}
\sum_{\alpha \in \fq} r^{+\alpha}(0) \prob\left(\pi=\rv^{+\alpha}|\pi \in c(\rv)\right) & = & \sum_{\alpha \in \fq } r(\alpha) \frac{K r(\alpha)}{r(0)} = 
\sum_{\alpha \in \fq } r(\alpha)^2.
\end{eqnarray*}
This implies that 
\begin{eqnarray*}
\esp(\pi(0)) & = & \sum_{\rv \in {\mathcal R}} \prob(\pi \in c(\rv)) \sum_{\alpha \in \fq} r^{+\alpha}(0) \prob\left(\pi=\rv^{+\alpha}|\pi \in c(\rv)\right) \\
& = & \sum_{\rv \in {\mathcal R}} \prob(\pi \in c(\rv))   \sum_{\alpha \in \fq } r(\alpha)^2 \\
& = & \sum_{\rv \in {\mathcal R}} \prob(\pi \in c(\rv)) \norm{\rv}^2  =  \esp\left( \norm{\pi}^2\right).
\end{eqnarray*}
\end{proof}

Let us recall the Chebyshev inequality which says that for any random variable $X$ we have
\begin{equation}\label{eq:second_moment}
\prob(|X- \esp(X)| \geq t) \leq \frac{\var(X)}{t^2}
\end{equation}
We are going to use this result with $X=\left\langle \Pi, \lfloor \mathbf 0 \rfloor \right\rangle$
and $X= \left\langle \Pi, \Pi \right\rangle$. This leads to the following concentration results.

\begin{lemma}\label{lem:concentration}
Let $\epsilon > 0$. We have
\begin{eqnarray}
\prob\left( \left\langle \Pi, \lfloor \mathbf 0 \rfloor \right\rangle \leq (1-\epsilon)  n \esp(\norm{\pi}^2)\right) 
&\leq& \frac{1}{n \epsilon^2  \left( \esp(\norm{\pi}^2 \right)^2}\label{eq:concentration_num}\\
\prob\left( \left\langle \Pi, \Pi \right\rangle \geq (1+\epsilon)  n \esp(\norm{\pi}^2)\right) 
&\leq& \frac{1}{n \epsilon^2  \left( \esp(\norm{\pi}^2 \right)^2}\label{eq:concentration_den}
\end{eqnarray}
\end{lemma}
\begin{proof}
Let us first prove \eqref{eq:concentration_num}. 
First observe that 
$$
\left\langle \Pi, \lfloor \mathbf 0 \rfloor \right\rangle = \sum_{i=1}^n \Pi^i(0)
$$
By linearity of expectation and Lemma \ref{lem:simple_formula} we have
\begin{equation}\label{eq:expectation1}
\esp\left\{ \left\langle \Pi, \lfloor \mathbf 0 \rfloor \right\rangle\right\} = n \esp(\pi(0)) = 
n \esp\left( \norm{\pi}^2 \right).
\end{equation}
Since the column vectors $\Pi^1(0),\Pi^2(0),\dots,\Pi^n(0)$ are independent random variables we also obtain
\begin{equation}
\label{eq:variance1}
\var\left\{ \left\langle \Pi, \lfloor \mathbf 0 \rfloor \right\rangle\right\}
= \var\left( \sum_{i=1}^n \Pi^i(0) \right)
= n \var\left( \pi(0) \right) \leq n.
\end{equation}
From this we deduce
\begin{eqnarray*}
\prob\left( \left\langle \Pi, \lfloor \mathbf 0 \rfloor \right\rangle \leq (1-\epsilon)  n \esp(\norm{\pi}^2)\right) 
&\leq& \prob\left( \left| \left\langle \Pi, \lfloor \mathbf 0 \rfloor \right\rangle - n \esp(\norm{\pi}^2) \right| \geq \epsilon   n \esp(\norm{\pi}^2)\right) \text{  }\\
& \leq & \frac{\var\left\{ \left\langle \Pi, \lfloor \mathbf 0 \rfloor \right\rangle\right\}}{\epsilon^2 n^2 \left(\esp(\norm{\pi}^2)\right)^2}
\text{  by \eqref{eq:second_moment} and  \eqref{eq:expectation1}}\\
& \leq & \frac{1}{\epsilon^2 n \left(\esp(\norm{\pi}^2)\right)^2}
\text{  by \eqref{eq:variance1}}
\end{eqnarray*}
This proves \eqref{eq:concentration_num}. The second statement follows by similar considerations. We have in this case
\begin{eqnarray}
\esp\left\{ \left\langle \Pi, \Pi \right\rangle\right\} & = &  n \esp\left( \norm(\pi)^2 \right) \label{eq:expectation2}\\
\var\left\{ \left\langle \Pi, \Pi \right\rangle \right\} & \leq & n\label{eq:variance2}
\end{eqnarray}
This can be used to prove that
\begin{eqnarray*}
\prob\left( \left\langle \Pi, \Pi \right\rangle \geq (1+\epsilon)  n \esp(\norm{\pi}^2)\right) 
&\leq& \prob\left( \left| \left\langle \Pi, \Pi \right\rangle - n \esp(\norm{\pi}^2) \right| \geq \epsilon   n \esp(\norm{\pi}^2)\right) \text{  }\\
& \leq & \frac{\var\left\{ \left\langle \Pi, \Pi \right\rangle\right\}}{\epsilon^2 n^2 \left(\esp(\norm{\pi}^2)\right)^2}
\text{  by \eqref{eq:second_moment} and \eqref{eq:expectation2}}\\
& \leq & \frac{1}{\epsilon^2 n \left(\esp(\norm{\pi}^2)\right)^2}
\text{  by \eqref{eq:variance2}}
\end{eqnarray*}
\end{proof}

We are ready now to prove the following theorem which gives a (tight) lower bound on the 
error-correction capacity of the Koetter-Vardy decoding algorithm over a discrete memoryless channel.
\begin{theorem}\label{th:KVP}
Let $(\code{C}_n)_{n \geq 1}$ be an infinite family of Reed-Solomon codes of rate $\leq R$. Denote by $q_n$ the alphabet size
of $\code{C}_n$ that is assumed to be a non decreasing sequence that goes to infinity with $n$.
Consider an infinite family of $q_n$-ary symmetric channels with associated probability error vectors $\pi_n$ such 
that $\esp\left( \norm{\pi_n}^2 \right)$ has a limit as $n$ tends to infinity.
Let 
$$
C_{\text{KV}} \eqdef \lim_{n \rightarrow \infty}  \esp\left( \norm{\pi_n}^2 \right).
$$
This infinite family of codes  can be decoded correctly 
by the Koetter-Vardy decoding algorithm with probability $1-o(1)$ 
as $n$ tends to infinity as soon as there exists 
$\epsilon >0$ such that 
$$
R \leq  C_{\text{KV}} -\epsilon.
$$
\end{theorem}

\begin{remark}
Let us observe that for the $q\hbox{-SC}_{p}$ we have
$$
\esp\left( \norm{\pi}^2 \right) = (1-p)^2 +(q-1)\frac{p^2}{(q-1)^2} = (1-p)^2 + \Oq.
$$ 
By letting $q$ going to infinity, we recover in this way the performance of the Guruswami-Sudan algorithm which works as soon
as 
$R < (1-p)^2$.
\end{remark}

\begin{proof}[Proof of Theorem \ref{th:KVP}]
Without loss of generality we may assume that the codeword that was sent is the zero codeword $\mathbf 0$.
Let $n_0$ be such that 
\begin{equation}\label{eq:bounds}
C_{\text{KV}} - \frac{\epsilon}{2} \leq \esp\left( \norm{\pi_n}^2 \right) \leq C_{\text{KV}} + \frac{\epsilon}{2}
\end{equation}
 for any $n \geq n_0$.
For $n \geq n_0$ we can write that the rate $R_n$ of $\Code{C}_n$
satisfies
\begin{eqnarray}
R_n & \leq & C_{\text{KV}} - \epsilon \nonumber \\
& \leq &   \esp\left( \norm{\pi_n}^2 \right) - \epsilon/2 \label{eq:bound} 
\end{eqnarray}
Let $K$ and $N$ be the dimension and the length of $\code{C}_n$.
In such a case we have
\begin{eqnarray}
\sqrt{\frac{K-1}{N}} &\leq & \sqrt{\frac{K}{N}} \nonumber\\
& \leq & \sqrt{\esp\left( \norm{\pi_n}^2 \right) - \epsilon/2} \;\;\;\text{(by \eqref{eq:bound})} \nonumber\\
& \leq & \sqrt{\esp\left( \norm{\pi_n}^2 \right)} \sqrt{1 - \frac{\epsilon}{2 \esp\left( \norm{\pi_n}^2 \right)}} \nonumber\\
& \leq &  \sqrt{\esp\left( \norm{\pi_n}^2 \right)}  \left( 1 - \frac{\epsilon}{4 \esp\left( \norm{\pi_n}^2 \right)}\right) 
\;\;\text{(since $\sqrt{1-x} \leq 1-\frac{x}{2}$.)}\nonumber\\
& \leq & \sqrt{\esp\left( \norm{\pi_n}^2 \right)}  \left( 1 - \frac{\epsilon}{4 C_{\text{KV}}  + 2\epsilon}\right) \label{eq:final}
\;\;\text{(by \eqref{eq:bounds})}
\end{eqnarray}
Let $\delta$ be a positive constant that we are going to choose afterwards.
Note that if an event $\mathcal E_1$ has probability $\geq 1 - \epsilon_1$ 
and another event $\mathcal E_2$ has probability $\geq 1 - \epsilon_2$, 
then 
$$\prob(\mathcal E_1 \cap \mathcal E_2) = \prob(\mathcal E_1) + 
\prob(\mathcal E_2) - \prob(\mathcal E_1 \cup \mathcal E_2) \geq  1 - \epsilon_1 + 1 - \epsilon_2 -1
= 1-\epsilon_1-\epsilon_2.$$
We can use this remark together with  Lemma \ref{lem:concentration} to deduce that with probability greater than or equal to
$1 - \frac{2}{N \delta^2  \left( \esp(\norm{\pi_n}^2 \right)^2}$ we have at the same time 
\begin{eqnarray}
\left\langle \Pi_n, \lfloor \mathbf 0 \rfloor \right\rangle & \geq & (1-\delta)  N \esp(\norm{\pi_n}^2) \label{eq:cond1}\\ 
 \left\langle \Pi_n, \Pi_n \right\rangle & \leq & (1+\delta)  N \esp(\norm{\pi_n}^2) \label{eq:cond2}
\end{eqnarray}
In such a case we have
\begin{equation}\label{eq:final2}
\frac{\left\langle \Pi_n, \lfloor \mathbf 0 \rfloor \right\rangle}{\sqrt{ \left\langle \Pi_n, \Pi_n \right\rangle}}
\geq \frac{1-\delta}{\sqrt{1+\delta}} \sqrt{N} \sqrt{\esp(\norm{\pi_n}^2)}
\end{equation}
There exists $x_0 >0$ such that for every $x \in [0,x_0]$ we have
$$
\frac{1-x}{\sqrt{1+x}} \leq 1-2x.
$$
Therefore for $\delta  \leq x_0$ we have in the aforementioned case
\begin{equation}
\frac{\left\langle \Pi_n, \lfloor \mathbf 0 \rfloor \right\rangle}
{\sqrt{\left\langle \Pi_n, \Pi_n \right\rangle}} \geq (1-2\delta) \sqrt{N} \sqrt{\esp(\norm{\pi_n}^2)}
\end{equation}
Let us choose now $\delta$ such that 
$$
\delta = \text{min}\left( x_0,\frac{\epsilon'}{ 4} \right).
$$
where $\epsilon' \eqdef \frac{\epsilon}{4 C_{\text{KV}}  + 2\epsilon}$. 
This choice implies 
\begin{eqnarray}
\frac{\left\langle \Pi_n, \lfloor \mathbf 0 \rfloor \right\rangle}
{\sqrt{\left\langle \Pi_n, \Pi_n \right\rangle}} & \geq &(1-\epsilon'/2) \sqrt{N} \sqrt{\esp(\norm{\pi_n}^2)} \nonumber\\
& \geq & \frac{1- \epsilon'/2}{1-\epsilon'} \sqrt{\frac{K-1}{N}}\label{eq:final3}
\end{eqnarray}
where we used \eqref{eq:final} for the last inequality.
Therefore we deduce that in the aforementioned case
(i.e. when \eqref{eq:cond1} and \eqref{eq:cond2} both hold), that we can choose $s$ appropriately in the Koetter-Vardy algorithm so that 
the codeword $\mathbf 0$ is in the list output by the algorithm.
The probability that \eqref{eq:final3} is satisfied is greater than or equal to 
$1 - \frac{2}{N \delta^2  \left( \esp(\norm{\pi_n}^2 \right)^2}$ which is also 
greater than or equal to (by using \eqref{eq:bounds})
$1 - \frac{2}{N \delta^2  \left( C_{\text{KV}} - \frac{\epsilon}{2} \right)^2}
$ which goes to $1$ as $N$ goes to infinity.
\end{proof}

\section{Correcting errors beyond the Guruswami-Sudan bound}
\label{Section4}
\subsection{The $\UV$-construction}
Now suppose we choose $U$ and $V$ as RS codes in a $\UV$ construction.
We start with a $q$-ary symmetric channel with error probability $p$.
Recall that the reliability matrix for the $U$-decoder is $\Pi_1=\Pi\times \Pi$
whereas for the $V$-decoder it is  $\Pi_2=\Pi \oplus \Pi$.
\begin{lemma}
\label{Lemma-UV}
Let $\pi_U$ and $\pi_V$ be the probability vectors corresponding to decoding the codes $U$ and $V$ respectively. 
\begin{itemize}
\item The channel error model of the code $V$ is a $q\hbox{-SC}_{p'}$ with $p'=2p-p^2$ and
$$\esp\left( \norm{\pi_{V}}^2 \right) = (1-p)^4 + \Oq.$$
\item For the channel error model of the code $U$ we have
$$\esp\left( \norm{\pi_{U}}^2 \right) =\frac{(p+2)(p-1)^2}{2-p} +  \Oq.$$
\end{itemize}
\end{lemma}

\begin{proof}
The proof of this Lemma can be found in Appendix \ref{Appendix-A}
\end{proof}

\begin{proposition}
As $q$ tends to infinity, the $\UV$-construction can be decoded correctly by the Koetter-Vardy decoding algorithm with probability $1-o(1)$ if 
$$R<\frac{(p^3-4p^2+4p-4)(1-p)^2}{2(p-2)}$$
\end{proposition}

\begin{proof}
The $\UV$-construction can be decoded correctly by the Koetter-Vardy decoding algorithm if
it decodes correctly $U$ and $V$. By Theorem \ref{th:KVP} decoding succeeds with probability $1-o(1)$ when we choose the rate $R_U$ of $U$ to be any 
positive number below $\esp\left( \norm{\pi_{U}}^2 \right)$ and the rate of $V$ any positive number below
$\esp\left( \norm{\pi_{V}}^2 \right)$. Since the rate $R$ of the $\UV$ construction is equal to $\frac{R_U+R_V}{2}$ 
decoding succeeds if
$$R< \lim_{q \rightarrow \infty} \frac{
\esp \left\{\norm{\pi_{U}}^2 \right\}  
+ \esp \left\{\norm{\pi_{V}}^2 \right\}  
}{2} 
= \frac{(p^3-4p^2+4p-4)(1-p)^2}{2(p-2)}.$$
\end{proof}

From Figure \ref{TwiceUV} we deduce that the $\UV$ decoder outperforms the RS decoder with
Guruswami-Sudan 
as soon as $R<0.168$.

\subsection{Recursive application of the $\UV$ construction}
Now we will study what happens over the $q\hbox{-SC}_{p}$ if we apply recursively the $\UV$ construction.  So we start with a $\UV$ code, we choose $U$ to be a $\gUV{U_1}{V_1}$ code and $V$ to be a $\gUV{U_2}{V_2}$ code, where $U_1$, $U_2$, $V_1$ and $V_2$ are RS codes over the same alphabet $\mathbb F_q$ and  of the same length. In other words, we look for a code of the form
\begin{eqnarray*}
\AB =   \\
\left\{(\uv_1| \uv_1+\vv_1|\uv_1+\uv_2| \uv_1 + \uv_2 + \vv_1+\vv_2): \uv_i \in U_i,\vv_i \in V_i \right\}\nonumber 
\end{eqnarray*}
From Lemma \ref{Lemma-UV} we obtain 
the channel error models for decoding 
$U_1$, $V_1$, $U_2$ and $V_2$ respectively, their reliability matrices are given by
$\Pi_1 \times \Pi_1$, $\Pi_1 \oplus \Pi_1$, $\Pi_2 \times \Pi_2$ and $\Pi_2 \oplus \Pi_2$ respectively (see Fig. \ref{fig:channel models}). We let $p' \eqdef 2p-p^2$.
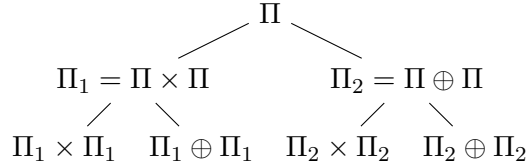
\begin{figure}[h!]
\begin{center}
\begin{tikzpicture}[scale=.6]
\tikzstyle{level 1}=[sibling distance=60mm] 
\tikzstyle{level 2}=[sibling distance=30mm]
\node (z){$\Pi$}
  child {node (a) {$\Pi_1=\Pi \times \Pi$}
    child {node (b) {$\Pi_1\times \Pi_1$}}
    child {node (g) {$\Pi_1\oplus \Pi_1$}}
  }
  child {node (j) {$\Pi_2=\Pi \oplus \Pi$}
    child {node (k) {$\Pi_2\times  \Pi_2$}}
    child {node (l) {$\Pi_2 \oplus \Pi_2$}}
   };
\end{tikzpicture}
\caption{The channel error models for $\AB$. \label{fig:channel models}}
\end{center}
\end{figure}

\begin{lemma}
\label{lem:un}
Let $\pi_{U_i}$ and $\pi_{V_i}$ be the probability vectors corresponding to decoding the codes $U_i$'s and $V_i$'s.
\begin{itemize}
\item
The channel error model of the code $V_2$ is a $q\hbox{-SC}_{p''}$ with 
$p''=2p'-{p'}^2$ and
$$\esp\left( \norm{\pi_{V_2}}^2 \right) = (1-p)^8 + \Oq;$$
\item $\esp\left( \norm{\pi_{U_2}}^2 \right) = \frac{(2+p')(1-p')^2}{(2-p')} +  \Oq$;
\item $\esp\left( \norm{\pi_{V_1}}^2 \right)=(1-p)^4\left(\frac{2+3p+8p^2-4p^3}{2-p}\right) +   \Oq$;
\item $\esp\left( \norm{\pi_{U_1}}^2 \right)= \frac{(5p^3-6p^2-5p-4)(1-p)^2}{4-3p} +   \Oq$.
\end{itemize}
\end{lemma}

\begin{proof}
The proof of this Lemma can be found in  Appendix \ref{Appendix-B}
\end{proof}

\begin{proposition}
As $q$ tends to infinity, the $\AB$-construction can be decoded correctly by the Koetter-Vardy decoding algorithm with probability $1-o(1)$ if
\begingroup\makeatletter\def\f@size{10}\check@mathfonts
\def\maketag@@@#1{\hbox{\m@th\large\normalfont#1}}
$$
R< 
\frac{(3p^{10} - 34p^9 + 187p^8 - 628p^7 + 1376p^6 - 2016p^5 +
1970p^4 - 1272p^3 + 568p^2 - 208p + 64)(p - 1)^2}
{4(p^2 - 2p+2)(3p - 4)(p - 2)}
$$
\endgroup
\end{proposition}

\begin{proof}
Decoding of $\AB$ succeeds if the Koetter-Vardy decoder is able to decode correctly $U_1$, $U_2$, $V_1$ and $V_2$.
This happens with probability $1-o(1)$ as soon as the rates $R_{U_1},R_{U_2},R_{V_1}$ and $R_{V_2}$ of these codes satisfy for some
$\epsilon >0$
\begin{eqnarray*}
R_{U_1} & \leq & \esp\left( \norm{\pi_{U_1}}^2\right) -\epsilon \\
R_{U_2} & \leq & \esp\left( \norm{\pi_{U_2}}^2\right) -\epsilon \\
R_{V_1} & \leq & \esp\left( \norm{\pi_{V_1}}^2\right) -\epsilon \\
R_{V_1} & \leq & \esp\left( \norm{\pi_{V_2}}^2\right) -\epsilon
\end{eqnarray*}
Since the rate $R$ of $\AB$ is given by $$
R = \frac{R_{U_1} + R_{U_2}+ R_{V_1}+R_{V_2}}{4}
$$ we
finally obtain that decoding succeeds with probability $1-o(1)$ as soon as the rate
$R$ is chosen such that
\begin{align*}
R &<  \lim_{q \rightarrow \infty} \frac{\sum_{i=1}^2
\esp \left\{\norm{\pi_{U_i}}^2 \right\}  
+ \esp \left\{\norm{\pi_{V_i}}^2 \right\}  
}{4}\\
\end{align*}
This implies the proposition by plugging the value of these expecations by using Lemma \ref{lem:un}.
\end{proof}

\begin{figure}[h!]
\includegraphics[width=1\linewidth]{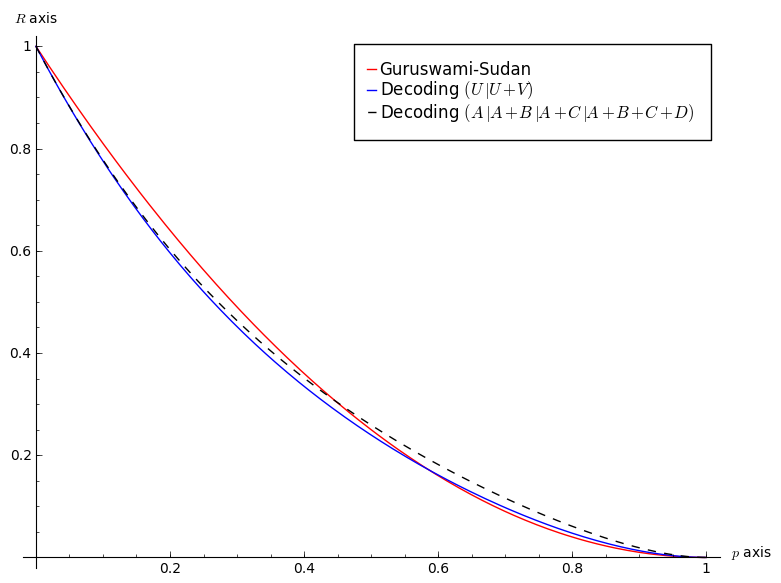}
\caption{Rate  plotted against the crossover error probability $p$ for several algorithms. The red line refers to the Guruswami-Sudan algorithm, the blue line to the $\UV$-construction and the dashed line to 
the $\AB$-construction.}
\label{TwiceUV}
\end{figure}

From Figure \ref{TwiceUV} we deduce that if we apply twice the $\UV$-construction we get better performance than decoding a classical RS code with the Guruswami-Sudan decoder for low rate codes, specifically for $R<0.3$.

\section{A new Mc-Eliece scheme}
As we have seen, this $\UV$ construction gives codes which in the low rate regime have even better 
error correction capacities than a standard RS code. This suggests to use such codes in a McEliece cryptosystem 
 to replace the original Goppa codes. These $\UV$ codes do not only have a better error correction capacity, they also allow to avoid  the Sidelnikov-Shestakov attack \cite{SS92} that broke
a previous proposal based on GRS codes \cite{N86}. 
Furthermore we can even strengthen the security of this scheme by using instead of the $\UV$ construction a generalized $\UV$
code which has trivially the same error-correction capacity as the $\UV$ construction but with better minimum distance properties 
which seems essential to avoid attacks based on finding minimum weight codewords in the code and trying to 
unravel the code structure from those minimum weight codewords. Analyzing precisely attacks of this kind needs however 
additional tools due to the peculiar structure of these generalized $\UV$ codes (it is for instance inappropriate to use the analysis done for random codes) and is out of scope of this paper. 


\begin{definition}
\label{Gen-UV}
Let $(U,V)$ be a pair of codes with parameters $[n,k_u,d_u]_q$ and $[n,k_v,d_v]_q$, respectively. Consider the
following matrix 
$$\Dm=\left(\begin{array}{c|c}
\Dm_1 & \Dm_3\\
\hline
\Dm_2 & \Dm_4
\end{array}\right)
\in \mathbb F_q^{n\times n}$$
where the $\Dm_i$'s are diagonal matrices
such that $\Dm$ is non singular.
We define the generalized $\UV$-construction of $U$ and $V$ with respect to $\Dm$ as the matrix product code:
$$\left\{ \guv \mid \mathbf u\in U \hbox{ and } \mathbf v\in V \right\}.$$
It is denoted by  $[U,V]\cdot \Dm$.
\end{definition}

\begin{remark}
\label{Gen-UV}
Let $U$ and $V$ be codes with generator matrices $\Gm_u$ and $\Gm_v$, and parity check matrices $\Hm_u$ and $\Hm_v$, respectively. 
It is a simple exercise to show that 
$$
\Gm = \left(\begin{array}{c|c} 
\Gm_u \Dm_1 & \Gm_u  \Dm_3\\
\hline
\Gm_v \Dm_2 & \Gm_v \Dm_4
\end{array}\right) 
~~\hbox{ and }~~ 
\Hm = \left( 
\begin{array}{c|c} 
\Hm_u \Dm_1& -\Hm_u \Dm_3\\
\hline
\Hm_v \Dm_2 & -\Hm_v \Dm_4
\end{array}\right)
$$
is a generator matrix and a parity check matrix, respectively 
for $[U,V]\cdot \Dm$.
\end{remark}

We consider the matrix-product construction $[U,V]\cdot \Dm$ which was already introduced in \cite{BZ74} and rediscovered in \cite{BN01,OS02}. In \cite[Theorem3.7]{BN01} a lower bound for the minimum distance of such code is given when the matrix $\Dm$ has a certain property, namely non-singular by columns. 
In \cite{OS02} a similar result is proved but makes the hypothesis that $U$ contains $V$ (but $\Dm$ is arbitrary). 
Their result does not seem to cover exactly our case, therefore we give a proof below.

\begin{lemma}
\label{Lemma-minimumDistance}
The code $[U,V]\cdot \Dm$ has parameters
$[2n,k_u+k_v,d]$ with $$\min\{2d_u,d_v\}\leq d \leq \min\{2d_u,2dv\}.$$
\end{lemma}

\begin{proof}
It is clear that $\mathcal C$ has length $2n$ and dimension $k_u+k_v$.

Now, consider a nonzero codeword $\mathbf c = \guv \in \mathcal C$. 
We distinguish two cases:
\begin{itemize}
\item If $\mathbf v=0$. Then,
$\mathrm w_H(\mathbf c) = 2 \mathrm w_H(\mathbf u) \geq 2d_u$.
 \item Otherwise, if $\mathbf v\neq 0$. By the Triangle Inequality $\mathrm w_H(a+b)\geq \mathrm w_H(a)-\mathrm w_H(b)$ and the fact that $D$ is non singular we have that
\begin{eqnarray*}
\mathrm w_H(\mathbf c) & = & 
\mathrm w_H(\mathbf u \Dm_1 + \mathbf v \Dm_2) +
\mathrm w_H(\mathbf u \Dm_3 + \mathbf v \Dm_4) \\
& \geq & \mathrm w_H(\mathbf u \Dm_1 + \mathbf v \Dm_2) +
\mathrm w_H(\mathbf v\left( \Dm_4 - \Dm_2 \Dm_1^{-1} \Dm_3 \right)) - 
\mathrm w_H((\mathbf u \Dm_1 + \mathbf v\Dm_2)\Dm_1^{-1}\Dm_3)\\
& = & \mathrm w_H(\mathbf v) \geq d_v
\end{eqnarray*}
\end{itemize}
Thus, $d\geq \min \{2d_u, d_v\}$. Moreover, take $\mathbf u\in U$ and $\mathbf v = 0$ with $\mathrm w_H(\mathbf u)=d_u$.
In such a case $w_H(\guv) = 2 d_u$ and therefore $d \leq 2 d_u$.
The other upper bound follows by choosing $\mathbf V\in V$ and $\mathbf u = 0$ with $\mathrm w_H(\mathbf v)=d_v$.
\end{proof}

Note that the minimum distance of this generalized $\UV$ construction can supersede the minimum distance of the standard $\UV$ construction which is 
equal to $\min\{2d_u,d_v\}$.

\begin{lemma}
\label{Dual-Lemma-minimumDistance}
The dual code of $[U,V]\cdot \Dm$ is the matrix product code $\mathcal [U^{\perp}, V^{\perp}]\cdot \Dm'$ with 
$$D'=\left(\begin{array}{c|c}
\Dm_1 & -\Dm_3\\
\hline
\Dm_2 & -\Dm_4
\end{array}\right)$$
This code has parameters
$[2n,n-(k_u+k_v),d^{\perp}]$ with $\min\{2d_u^{\perp},d_v^{\perp}\}\leq d^{\perp} \leq \min\{2d_u^{\perp},2dv^{\perp}\}$. 
\end{lemma}

\begin{proof}
Remark \ref{Gen-UV} showed that the dual code of $[U,V]\cdot \Dm$ is a generalized $\UV$-construction of $U^{\perp}$ and $V^{\perp}$. Then, the result follows from Lemma \ref{Lemma-minimumDistance}.
\end{proof}

These generalized $\UV$ codes based on RS constituent codes have clearly an efficient decoding which is similar to the  $\UV$-decoder. There are only a few differences: when we receive a word $(\yv_1,\yv_2)$ we just compute
the difference $ \yv_1\Dm_3 - \yv_2 \Dm_1$ which should be a noisy version of $\vv(\Dm_2 \Dm_3  -  \Dm_4\Dm_1)$. 
However the error correction capacity is the same as the original $\UV$ with this kind of decoding algorithm.
More precisely, the McEliece scheme we propose is the following

\begin{itemize}
\item[] \textbf{Key generation:} 
\begin{itemize}
\item Choose $U, V$ as RS codes of some length $n$.
\item  Construct a random matrix $\Dm$ as described in Definition \ref{Gen-UV}.
\item Let $G$ be a random generator matrix of the code $\mathcal C= [U,V]\cdot \Dm \cdot {\mathbf \Sigma}_{2n}$ where
${\mathbf \Sigma}_{2n}$ is a permutation matrix of size $2n$ and $\mathcal A_{\mathcal C}$ a decoding algorithm for $\mathcal C$  that  typically corrects  $t$ errors. It consists in applying ${\mathbf \Sigma}_{2n}^{-1}$ to the received word and then performing
the aforementioned generalized $\UV$-decoder.
\end{itemize}
\item[]The \emph{public key} and the \emph{private key} are given respectively by:
$$\begin{array}{ccc}
\mathcal K_{\mathrm{pub}}=(G,t) & \hbox{ and }&
\mathcal K_{\mathrm{secret}}=\mathcal A_{\mathcal C}
\end{array}$$ 
\item[] \textbf{Encryption:} $\mathbf y = \mathbf m G + \mathbf e$ where $\mathbf m$ is the message and $\mathbf e$ is a random error vector of weight at most $t$.
\item[] \textbf{Decryption:} Use $\mathcal K_{\mathrm{secret}}$ to retrieve $\mathbf m$.
\end{itemize}

Thus, by choosing code $U$ and $V$ of large enough minimum distance we avoid attacks that try to recover the code structure by looking for low weight codewords either in the code or in its dual. If the minimum distance of the generalized $\UV$ code is equal to $2 d_u$ such codewords arise as
codewords of the form $(\uv \Dm_1|\uv \Dm_3) \Sigma_{2n}$. This clearly leaks information about $\Sigma_{2n}$ and $\Dm_1$ and 
$\Dm_3$ if we are able
to find such codewords. This is why we want to avoid that such codewords can be easily found.

Note that Wang proposed in \cite{W15} a very similar scheme, with the difference that $U$ was a random code and $V$ a RS code and 
that he took only a subcode of the generalized $\UV$ code namely the code generated by 
$\left(\begin{array}{c|c} 
\Gm_u \Dm_1 + \Gm_v \Dm_2 & \Gm_u  \Dm_3 + \Gm_v \Dm_4\end{array}\right)$. The code rate loss implied by this choice  results in 
a significant loss in the key size (since we have to protect ourself against generic decoders for $t$ errors  for a code which is of much smaller dimension). The fact that $U$ is random in his scheme however is a rather strong argument in favor of its security.

\bibliographystyle{myabbrv}
\bibliography{u+v-arxiv}

\appendix
\section{The $\UV$ construction}
\label{Appendix-A}
Along the following two sections and by abuse of notation we will use the symbol $\oq$ not only to describes the big O notation but also as a vector whose $L$-infinity norm is behaved like the function $\frac{1}{q}$ when $q$ tends to infinity, in other words, a vector whose elements tends to zero as $q$ goes to infinity . Recall that the infinity norm of a vector $\mathbf v\in \mathbb F_q^n$, denoted $\norm{\mathbf v}_{\infty}$, is defined as the maximum of the absolute values of its components, i.e.
$$\norm{\mathbf v}_{\infty} = \max \left\{|v(i) \mid i = 1, \ldots ,n | \right\}.$$

In this appendix, $U$ and $V$ are RS codes and we will use them in a $\UV$ construction.
We start with a $q$-ary symmetric channel with error probability $p$ ($q\hbox{-SC}_{p}$). 
On the following we obtain the channel error models for decoding $U$ and V.
Recall that the reliability matrix for the $U$-decoder is $\Pi_1=\Pi\times \Pi$
whereas for the $V$-decoder it is  $\Pi_2=\Pi \oplus \Pi$.

Suppose we transmit the codeword $\uv$ over a noisy channel and we receive the vector
$$\mathbf y = 
\left(\mathbf y_1 \mid \mathbf y_2\right) = 
\uv + \left( \mathbf e_1 \mid \mathbf e_2\right)$$
We begin to observe that in the case of a $q$-ary symmetric channel we have only the following possibilities
\begin{enumerate}[ncases]
\item \label{it:1}No error has occurred in position $i$: $e_1(i) = e_2(i) = 0$. 
\item \label{it:2}An error has occured in position $i$: $e_1(i)\neq 0$ or $e_2(i)\neq 0$. 
\item \label{it:3}Two errors have occurred in position $i$: $e_1(i) \neq 0$ and $e_2(i)\neq 0$. 
\end{enumerate}

\begin{table}[h!]
\centering
\begin{tabular}{|c|c|}
\cline{2-2}
\multicolumn{1}{c|}{} & \begin{tabular}{c}Probability of \\ occurrence\end{tabular}\\
\hline
\begin{tabular}{c}
\textbf{Case} $1$\\
(No errors)
\end{tabular} & $(1-p)^2$\\
\hline
\begin{tabular}{c}
\textbf{Case} $2$\\
($1$ error)
\end{tabular} & $2p(1-p)$\\
\hline
\begin{tabular}{c}
\textbf{Case} $3$\\
($2$ errors)
\end{tabular} & $p^2$\\
\hline
\end{tabular}
\caption{Ways of combining the vectors $y_1(i)$ and $y_2(i)$}
\end{table}

\subsection{The matrix $\Pi_1=\Pi \times \Pi$ in the $q$-ary symmetric channel}

\begin{lemma}
\label{L12}
Let $\pi_U$ be the probability vector corresponding to decoding the code $U$. For the channel error model of the code $U$ we have
\begin{eqnarray*}
\esp\left( \norm{\pi_{U}}^2 \right) & = & (1-p)^2 +2 p (2-p)\left(\frac{1-p}{2-p}\right)^2 +  \Oq =  \frac{(p+2)(p-1)^2}{2-p} +  \Oq
\end{eqnarray*}
\end{lemma}

\begin{proof}
We will treat each case as a separate study.

\begin{enumerate}[ncases]
\item We have (up to permutation)
$$\Pi_{\mathbf y_1}^i \times \Pi_{\mathbf y_2}^i = \left(\begin{array}{c}
1-p \\ \frac{p}{q-1} \\ \vdots \\ \frac{p}{q-1}
\end{array}\right) \times \left(\begin{array}{c}
1-p \\ \frac{p}{q-1} \\ \vdots \\ \frac{p}{q-1}
\end{array}\right) = 
\left(\begin{array}{c}
\frac{(1-p)^2}{\beta}\\\frac{p^2}{(q-1)\beta}\\ \vdots \\ \frac{p^2}{(q-1)\beta}
\end{array}\right)=
\left(\begin{array}{c}
1-s \\ \frac{s}{q-1} \\ \vdots \\ \frac{s}{q-1}
\end{array}\right) = 
\left(\begin{array}{c} 1 \\ 0 \\ \vdots \\ 0\end{array}\right) + \oq$$
with $\beta = (1-p)^2 + (q-1)\frac{p^2}{(q-1)^2}$.
And, 
$s=\frac{p^2}{(1-p)^2(q-1)+p^2} = \Oq$

\item We have (up to permutation)

\begin{eqnarray*}
\Pi_{\mathbf y_1}^i \times \Pi_{\mathbf y_2}^i & = &   \left(\begin{array}{c}
1-p \\ \frac{p}{q-1} \\  \frac{p}{q-1}  \\ \vdots \\ \frac{p}{q-1}
\end{array}\right) \times \left(\begin{array}{c}
\frac{p}{q-1} \\ 1-p \\ \frac{p}{q-1} \\ \vdots \\ \frac{p}{q-1}
\end{array}\right) = 
\left(\begin{array}{c}
\frac{(1-p)p}{(q-1)\beta} \\ \frac{(1-p)p}{(q-1)\beta} \\ \frac{p^2}{(q-1)^2 \beta} \\ \vdots \\ 
\frac{p^2}{(q-1)^2 \beta}
\end{array}\right) = 
\left(\begin{array}{c}
\frac{1-p}{2-p}\\ \frac{1-p}{2-p} \\ \frac{p}{(q-1)(2-p)} \\ \vdots \\  \frac{p}{(q-1)(2-p)}
\end{array}\right) + \oq
\end{eqnarray*}
with 
$$\beta = 2\frac{(1-p)p}{(q-1)}+(q-2)\frac{p^2}{(q-1)^2} = \frac{p(2-p)}{(q-1)} + \Oq$$

\item The probability that the same error occurred at $\mathbf e_1(i)$ and $\mathbf e_2(i)$ is assummed to be negligible. Hence (up to permutation),
\begin{eqnarray*}
\Pi_{\mathbf y_1}^i \times \Pi_{\mathbf y_2}^i & = &  \left(\begin{array}{c}
\frac{p}{q-1}  \\ 1-p \\ \frac{p}{q-1} \\ \frac{p}{q-1} \\ \vdots \\ \frac{p}{q-1}
\end{array}\right) \times \left(\begin{array}{c}
\frac{p}{q-1} \\ \frac{p}{q-1} \\ 1-p \\ \frac{p}{q-1} \\ \vdots \\ \frac{p}{q-1}
\end{array}\right)  = 
\left(\begin{array}{c}
\frac{1-p}{2-p}\\ \frac{1-p}{2-p} \\ \frac{p}{(q-1)(2-p)} \\ \vdots \\  \frac{p}{(q-1)(2-p)}
\end{array}\right) + \oq 
\end{eqnarray*}
\end{enumerate}
Once we know all the columns of matrix $\Pi \times \Pi$, the result follows easily.
\end{proof}
Hence, the channel error model of the code $U$ is represented by the reliability matrix $\Pi_1 = \Pi \times \Pi$ and the expectation of the $L_2$ norm $\norm{\pi}^2$ of a column $\pi$ of $\Pi_1$ is given by 
$$\esp\left( \norm{\pi_{U}}^2 \right)
(1-p)^2 +2 p (2-p)\left(\frac{1-p}{2-p}\right)^2 +  \Oq
$$

\subsection{The matrix $\Pi_2=\Pi \oplus \Pi$ in the $q$-ary symmetric channel}

\begin{lemma}
\label{L11}
Let $\pi_V$ be the probability vector corresponding to decoding the code $V$. 
The channel error model of the code $V$ is a $q\hbox{-SC}_{p'}$ with $p'=2p-p^2$ and
\begin{eqnarray*}
\esp\left( \norm{\pi_{V}}^2 \right) &=& (1-p')^2 + \Oq = (1-p)^4 + \Oq.
\end{eqnarray*}
\end{lemma}

\begin{proof}
We will treat each case as a separate study.

\begin{enumerate}[ncases]
\item No error occurred in position $i$, i.e. $\Pi_{\mathbf y_1}^i = \Pi_{\mathbf y_2}^i$. Hence (up to permutation), 
$$\Pi_{\mathbf y_1}^i \oplus \Pi_{\mathbf y_2}^i =  
\left(\begin{array}{c}
1-p \\ \frac{p}{q-1} \\ \vdots \\ \frac{p}{q-1}
\end{array}\right) \oplus \left(\begin{array}{c}
1-p \\ \frac{p}{q-1} \\ \vdots \\ \frac{p}{q-1}
\end{array}\right) = \left(\begin{array}{c}
1-p' \\ \frac{p'}{q-1}  \\ \vdots \\ \frac{p'}{q-1} 
\end{array}\right)  + \oq$$
with $p'=2p-p^2$

\item One error occurred in position $i$, in other words, the order of the elements of $\Pi_{\mathbf y_1}^i$ and $\Pi_{\mathbf y_2}^i$ are different. Thus (up to permutation), 
$$\Pi_{\mathbf y_1}^i \oplus \Pi_{\mathbf y_2}^i =  \left(\begin{array}{c}
1-p \\ \frac{p}{q-1} \\  \frac{p}{q-1}  \\ \vdots \\ \frac{p}{q-1}
\end{array}\right) \oplus \left(\begin{array}{c}
\frac{p}{q-1} \\ 1-p \\ \frac{p}{q-1} \\ \vdots \\ \frac{p}{q-1}
\end{array}\right) = \left(\begin{array}{c}
\frac{p'}{q-1} \\ 1-p' \\ \frac{p'}{q-1} \\ \vdots \\ \frac{p'}{q-1}
\end{array}\right) + \oq$$
with $p'=2p-p^2$
\item We will have two options: either $\mathbf e_1(i)=\mathbf e_2(i)$ (similar to \ref{it:1}, up to permutation)  or $\mathbf e_1(i)\neq \mathbf e_2(i)$ (similar to \ref{it:2}, up to permutation). 
\end{enumerate}

Thus, the transition matrix $\Pi \oplus \Pi$ can be  represented as a $q\hbox{-SC}_{p'}$ with $p'=2p-p^2$.
\end{proof}

%

\section{Recursive application of the $\UV$ construction}
\label{Appendix-B}
In this appendix we study what happens over a $q$-ary symmetric channel with error probability $p$ ($q\hbox{-SC}_{p}$) if we apply recursively the $\UV$ construction.  That is, we start with a $\UV$ code, we choose $U$ to be a $\gUV{U_1}{V_1}$ code and $V$ to be a $\gUV{U_2}{V_2}$ code, where $U_1$, $U_2$, $V_1$ and $V_2$ are RS codes over the same alphabet $\mathbb F_q$ and  of the same length. In other words, we look for a code of the form
\begin{gather}
\AB =  \nonumber\\
\left\{\left( \uv_1 \mid \uv_1 + \vv_1 \mid \uv_1 + \uv_2 \mid \uv_1 + \uv_2 + \vv_1 + \vv_2 \right)~:~\uv_i \in U_i,\vv_i \in V_i \right\}.\nonumber 
\end{gather}

From Lemma \ref{L11} and \ref{L12} we will obtain the channel error models for decoding  $U_1$, $V_1$, $U_2$ and $V_2$ respectively, their reliability matrices are given by
$\Pi_1 \times \Pi_1$, $\Pi_1 \oplus \Pi_1$, $\Pi_2 \times \Pi_2$ and $\Pi_2 \oplus \Pi_2$ respectively (see Figure \ref{fig:channel models}). We use the previous notation $p' \eqdef 2p-p^2$.

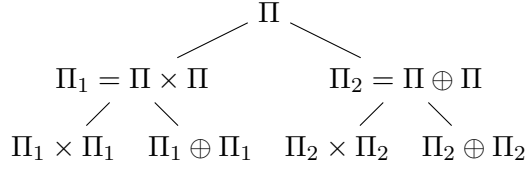
\begin{figure}[h!]
\begin{center}
\begin{tikzpicture}[scale=.6]
\tikzstyle{level 1}=[sibling distance=60mm] 
\tikzstyle{level 2}=[sibling distance=30mm]
\node (z){$\Pi$}
  child {node (a) {$\Pi_1=\Pi \times \Pi$}
    child {node (b) {$\Pi_1\times \Pi_1$}}
    child {node (g) {$\Pi_1\oplus \Pi_1$}}
  }
  child {node (j) {$\Pi_2=\Pi \oplus \Pi$}
    child {node (k) {$\Pi_2\times  \Pi_2$}}
    child {node (l) {$\Pi_2 \oplus \Pi_2$}}
   };
\end{tikzpicture}
\caption{The different channel error models for $\AB$\label{fig:channel models}}
\end{center}
\end{figure}

Suppose we transmit the codeword 
$$\left( \uv_1 \mid \uv_1 + \vv_1 \mid \uv_1 + \uv_2 \mid \uv_1 + \uv_2 + \vv_1 + \vv_2 \right)$$ 
over a noisy channel and we receive the vector
$$\mathbf y = 
\left(\mathbf y_1 \mid \mathbf y_2 \mid \mathbf y_3 \mid \mathbf y_4\right) = 
\left( \uv_1 \mid \uv_1 + \vv_1 \mid \uv_1 + \uv_2 \mid \uv_1 + \uv_2 + \vv_1 + \vv_2 \right) + 
\left( \mathbf e_1 \mid \mathbf e_2 \mid \mathbf e_3  \mid \mathbf e_4\right)$$
We begin to observe that in the $q$-ary symmetric channel we have only the possibilities given in Table \ref{Table2}

\begin{table}[h!]
\begin{center}
\begin{tabular}{|c|c|c|}
\cline{2-3}
\multicolumn{1}{c|}{} & Result of the combination of ... & 
\begin{tabular}{c}Probability of \\ occurrence\end{tabular}\\
\hline
\begin{tabular}{c}
\textbf{Case} $4$ \\
(No errors)
\end{tabular}
& 
\begin{tabular}{c}
\ref{it:1} and \ref{it:1}\\
$e_j(i) = 0$ for all $j\in \{1, 2, 3, 4\}$ 
\end{tabular}
& $(1-p)^4$\\
\hline
\begin{tabular}{c}
\textbf{Case} $5$\\
($1$ error)
\end{tabular} & 
\begin{tabular}{c}
\ref{it:1} and \ref{it:2}\\
$\exists j_1\in \{1,2,3,4\}$ s.t. $e_{j_1}(i) \neq 0$ \\
And $e_j(i)=0$, otherwise. 
\end{tabular}
& $4p(1-p)^3$\\
\hline
\begin{tabular}{c}
\textbf{Case} $6$\\
($2$ errors)
\end{tabular} & 
\begin{tabular}{c}
\ref{it:1} and \ref{it:3}\\
$\exists j_1\in \{1,3\}$ s.t.
$e_{j_1}(i), e_{j_1+1}(i) \neq 0$ \\
And $e_j(i)=0$, otherwise. 
\end{tabular}
& $2p^2(1-p)^2$\\
\hline
\begin{tabular}{c}
\textbf{Case} $7$\\
($2$ errors)
\end{tabular} & 
\begin{tabular}{c}
\ref{it:2} and \ref{it:2}\\
$\exists j_1\in \{1,2\}$ and $j_2\in \{3,4\}$ s.t.
$e_{j_1}(i), e_{j_2}(i) \neq 0$ \\
And $e_j(i)=0$, otherwise. 
\end{tabular}
& $4p^2(1-p)^2$\\
\hline
\begin{tabular}{c}
\textbf{Case} $8$\\
($3$ errors)
\end{tabular} & 
\begin{tabular}{c}
\ref{it:2} and \ref{it:3}\\
$\exists j\in \{1,2,3,4\}$ s.t. $e_{j_1}(i) = 0$ \\
And $e_j(i)\neq 0$, otherwise. 
\end{tabular}
& $4p^3(1-p)$\\
\hline
\begin{tabular}{c}
\textbf{Case} $9$\\
($4$ errors)
\end{tabular} & 
\begin{tabular}{c}
\ref{it:3} and \ref{it:3}\\
$e_j(i) \neq 0$ for all $j\in \{1, 2, 3, 4\}$ 
\end{tabular}
& $p^4$\\
\hline
\end{tabular}
\caption{Ways of combining the vectors $y_1(i)$, $y_2(i)$, $y_3(i)$ and $y_4(i)$}
\label{Table2}
\end{center}
\end{table}

\subsection{The matrix $\Pi_2\oplus \Pi_2$ in the $q$-ary symmetric channel}

\begin{lemma}
Let $\pi_{V_2}$ be the probability vector corresponding to decoding the code $V_2$. The channel error model of the code $V_2$ is a $q\hbox{-SC}_{p''}$ with 
$p''=2p'-{p'}^2$ and
$$\esp\left( \norm{\pi_{V_2}}^2 \right) = (1-p")^2 + \Oq = (1-p)^8 + \Oq$$
\end{lemma}

\begin{proof}
Direct consequence of Lemma \ref{L11}.
\end{proof}

\subsection{The matrix $\Pi_2\times \Pi_2$ in the $q$-ary symmetric channel}

\begin{lemma}
Let $\pi_{U_2}$ be the probability vector corresponding to decoding the code $U_2$. We have 
\begin{eqnarray*}
\esp\left( \norm{\pi_{U_2}}^2 \right) & = & \frac{(2+p')(1-p')^2}{(2-p')} +  \Oq
\end{eqnarray*}
\end{lemma}

\begin{proof}
Since the transition matrix $\Pi \oplus \Pi$ can be represented as a $q\hbox{-SC}_{p'}$ with $p'=2p-p^2$, this case can be treat similar to Lemma \ref{L12}.
\end{proof}

\subsection{The matrix $\Pi_1\oplus \Pi_1$ in the $q$-ary symmetric channel}

\begin{lemma}
Let $\pi_{V_1}$ be the probability vector corresponding to decoding the code $V_1$. We have 
$$\esp\left( \norm{\pi_{V_1}}^2 \right)=(1-p)^4\left(\frac{2+3p+8p^2-4p^3}{2-p}\right) +   \Oq$$
\end{lemma}

\begin{proof}
We will treat each case as a separate study.
\begin{enumerate}[ncases, start=4]
\item We have that the column $\left(\Pi_{\mathbf y_1}^i \times \Pi_{\mathbf y_2}^i\right) \oplus 
\left( \Pi_{\mathbf y_3}^i \times \Pi_{\mathbf y_4}^i\right)$ is (up to permutation)
$$\left(\left( \begin{array}{c}1 \\ 0\\ \vdots \\ 0 \end{array} \right) + \oq \right)
\oplus 
\left(\left(\begin{array}{c}1 \\ 0 \\ \vdots \\ 0 \end{array} \right) + \oq \right) 
+ \oq= 
\left( \begin{array}{c}1  \\ 0 \\ \vdots \\ 0 \end{array}\right) + \oq$$
\item We have that the column $\left(\Pi_{\mathbf y_1}^i \times \Pi_{\mathbf y_2}^i\right) \oplus 
\left( \Pi_{\mathbf y_3}^i \times \Pi_{\mathbf y_4}^i\right)$ is (up to permutation)
$$
\left(\left( \begin{array}{c}1 + 0 \\ 0\\ \vdots \\ 0\end{array} \right)  + \oq\right)
\oplus
\left(
\left( \begin{array}{c}
\frac{1-p}{2-p} \\ \frac{1-p}{2-p} \\ \frac{p}{(q-1)(2-p)}  \\ \vdots \\ \frac{p}{(q-1)(2-p)} 
\end{array}\right) + \oq\right)
= \left( \begin{array}{c}
\frac{1-p}{2-p} \\ \frac{1-p}{2-p}  \\ 0 \\ \vdots \\ 0
\end{array}\right) + \oq
$$
\item Identical result to the above.
\item We have that the column $\left(\Pi_{\mathbf y_1}^i \times \Pi_{\mathbf y_2}^i\right) \oplus 
\left( \Pi_{\mathbf y_3}^i \times \Pi_{\mathbf y_4}^i\right)$ is (up to permutation)
\begin{eqnarray*}
\left(
\left( \begin{array}{c}
\frac{1-p}{2-p} \\ \frac{1-p}{2-p}  \\ \frac{p}{(q-1)(2-p)} \\ \frac{p}{(q-1)(2-p)}  \\ \vdots \\ \frac{p}{(q-1)(2-p)} 
\end{array}\right) + \oq \right)\oplus
\left(\left( \begin{array}{c}
\frac{1-p}{2-p} \\ \frac{p}{(q-1)(2-p)} \\ \frac{1-p}{2-p}  \\ \frac{p}{(q-1)(2-p)}  \\ \vdots \\ \frac{p}{(q-1)(2-p)}
\end{array}\right) +\oq\right) & = &
\left( \begin{array}{c}
\left(\frac{1-p}{2-p}\right)^2  \\ \left(\frac{1-p}{2-p}\right)^2  \\ \left(\frac{1-p}{2-p}\right)^2  \\ \left(\frac{1-p}{2-p}\right)^2  \\ 0 \\ \vdots \\ 0
\end{array}\right) + \oq
\end{eqnarray*}
\item Identical result to the above.
\item Identical result to the above.
\end{enumerate}
Once we know all the columns of matrix $\Pi_2 \oplus \Pi_2$, the result follows easily.
\end{proof}

\subsection{The matrix $\Pi_1 \times \Pi_1$ in the $q$-ary symmetric channel}

\begin{lemma}
Let $\pi_{U_1}$ be the probability vector corresponding to decoding the code $U_1$. We have 
\begin{eqnarray*}
\esp\left( \norm{\pi_{U_1}}^2 \right) &= & \frac{(5p^3 - 6p^2-5p-4)(1-p)^2}{4-3p} + \Oq
\end{eqnarray*}
\end{lemma}

\begin{proof}
We study separately three different cases:
\begin{itemize}
\item \textbf{Case} $4$, \textbf{Case} $5$ and \textbf{Case} $6$. In all these cases the columns of the transition matrix have the same form (up to permutation).
$$\Pi_{\mathbf y_1}^i \times \Pi_{\mathbf y_2}^i \times \Pi_{\mathbf y_3}^i \times \Pi_{\mathbf y_4}^i  = \left(\begin{array}{c} 1  \\ 0 \\ \vdots \\ 0  \end{array}\right) + \oq 
$$
\item \textbf{Case} $7$ We have that the column $\Pi_{\mathbf y_1}^i \times \Pi_{\mathbf y_2}^i \times \Pi_{\mathbf y_3}^i \times \Pi_{\mathbf y_4}^i$ is (up to permutation)

\begin{eqnarray*}
\hspace{-0.5cm}
 \left( \left(\begin{array}{c} 
\frac{1-p}{2-p} \\ \frac{1-p}{2-p}  \\ \frac{p}{(q-1)(2-p)}  \\ \frac{p}{(q-1)(2-p)}  \\ \vdots \\\frac{p}{(q-1)(2-p)}  \end{array}\right) + \oq \right) \times 
\left( \left(\begin{array}{c} 
\frac{1-p}{2-p} \\ \frac{p}{(q-1)(2-p)}  \\  \frac{1-p}{2-p}  \\ \frac{p}{(q-1)(2-p)}  \\ \vdots \\\frac{p}{(q-1)(2-p)}  \end{array}\right)  + \oq \right)
& = & \left( \begin{array}{c}
\left(\frac{1-p}{2-p}\right)^2 \frac{1}{\alpha}  \\
\frac{(1-p)p}{(2-p)^2(q-1)} \frac{1}{\alpha} \\
\frac{(1-p)p}{(2-p)^2(q-1)} \frac{1}{\alpha} \\
\frac{p^2}{(q-1)^2(2-p)^2}\frac{1}{\alpha} \\
\vdots \\
\frac{p^2}{(q-1)^2(2-p)^2}\frac{1}{\alpha} \\
\end{array}\right)  + \oq \\
 & = & 
\left(\begin{array}{c} 1 \\ 0 \\ \vdots \\ 0  \end{array}\right) + \oq
\end{eqnarray*}

$$\hbox{ with }\alpha = \left( \frac{1-p}{2-p}\right)^2 + 2\frac{(1-p)p}{(2-p)^2(q-1)} + (q-3)\frac{p^2}{(q-1)^2(2-p)^2} = \left(\frac{1-p}{2-p}\right)^2 + \Oq$$
\item \textbf{Case} $8$ We have that the column $\Pi_{\mathbf y_1}^i \times \Pi_{\mathbf y_2}^i \times \Pi_{\mathbf y_3}^i \times \Pi_{\mathbf y_4}^i$ is (up to permutation)

\begin{eqnarray*}
\hspace{-0.5cm}
\left(\left(\begin{array}{c} 
\frac{1-p}{2-p} \\ \frac{1-p}{2-p} \\ \frac{p}{(q-1)(2-p)} \\ \frac{p}{(q-1)(2-p)}  \\  \frac{p}{(q-1)(2-p)}  \\  \vdots \\\frac{p}{(q-1)(2-p)}  \end{array}\right) + \oq \right)\times 
\left(\left(\begin{array}{c} 
 \frac{p}{(q-1)(2-p)} \\   \frac{p}{(q-1)(2-p)} \\ \frac{1-p}{2-p} \\ \frac{1-p}{2-p} \\ \frac{p}{(q-1)(2-p)} \\ \vdots \\\frac{p}{(q-1)(2-p)}  \end{array}\right) + \oq\right)  & = & 
\left( \begin{array}{c}
\frac{(1-p)p}{(2-p)^2(q-1)} \frac{1}{\alpha}  \\
\frac{(1-p)p}{(2-p)^2(q-1)} \frac{1}{\alpha}  \\
\frac{(1-p)p}{(2-p)^2(q-1)} \frac{1}{\alpha}  \\
\frac{(1-p)p}{(2-p)^2(q-1)} \frac{1}{\alpha}  \\
\frac{p^2}{(q-1)^2(2-p)^2}\frac{1}{\alpha} \\
\vdots \\
\frac{p^2}{(q-1)^2(2-p)^2}\frac{1}{\alpha} \\
\end{array}\right) + \oq \\
& = &  
\left(\begin{array}{c} 
\frac{1-p}{4-3p} \\ 
\frac{1-p}{4-3p} \\ 
\frac{1-p}{4-3p} \\ 
\frac{1-p}{4-3p} \\ 
0 \\ \vdots \\ 0  \end{array}\right) + \oq
\end{eqnarray*}
with $$\alpha = 4 \frac{(1-p)p}{(2-p)^2(q-1)} + (q-4)\frac{p^2}{(q-1)^2(2-p)^2} = 
\frac{p(4-3p)}{(2-p)^2(q-1)} + \Oq$$

\item \textbf{Case} $9$ Identical result to the above. 
\end{itemize}
Once we know all the columns of matrix $\Pi_2 \times \Pi_2$, the result follows easily.
\end{proof}

\end{document}